\documentclass{elsart}
\usepackage{ifpdf}
\usepackage{graphicx,amssymb,amsmath,lineno,epsfig,subfig}
\usepackage{subfig}

\usepackage{indentfirst}

\ifpdf
\usepackage[%
  pdftitle={Instructions for use of the document class
    elsart},%
  pdfauthor={Simon Pepping},%
  pdfsubject={The preprint document class elsart},%
  pdfkeywords={instructions for use, elsart, document class},%
  pdfstartview=FitH,%
  bookmarks=true,%
  bookmarksopen=true,%
  breaklinks=true,%
  colorlinks=true,%
  linkcolor=blue,anchorcolor=blue,%
  citecolor=blue,filecolor=blue,%
  menucolor=blue,pagecolor=blue,%
  urlcolor=blue]{hyperref}
\else
\usepackage[%
  breaklinks=true,%
  colorlinks=true,%
  linkcolor=blue,anchorcolor=blue,%
  citecolor=blue,filecolor=blue,%
  menucolor=blue,pagecolor=blue,%
  urlcolor=blue]{hyperref}
\fi

\renewcommand{\theequation}{\arabic{section}.\arabic{equation}}
\newtheorem{theorem}{Theorem}[section]
\newtheorem{remark}{Remark}[section]

\newtheorem{lemma}{Lemma}[section]
\newenvironment{proof}[1][Proof]{\noindent\emph{{#1.}} }{ \rule{0.5em}{0.5em}}
\usepackage{indentfirst}
\usepackage{setspace}

\begin{document}
\begin{frontmatter}
\title{Blow-up of solution of an initial boundary value problem for a generalized Camassa-Holm equation}
\author{Jiangbo Zhou\corauthref{cor}},
\corauth[cor]{Corresponding author. Tel.: +86-511-88969336; Fax:
+86-511-88969336.} \ead{zhoujiangbo@yahoo.cn}
\author{Lixin Tian}
\address{Nonlinear Scientific Research Center, Faculty of Science, Jiangsu
University, Zhenjiang, Jiangsu 212013, China}
\begin{abstract} In this paper, we study the following initial
boundary value problem for a generalized Camassa-Holm equation
\[
\left\{ {{\begin{array}{*{20}c}
 {u_t - u_{xxt} + 3uu_x - 2u_x u_{xx} - uu_{xxx} + k(u - u_{xx} )_x=
0,  t\geq0,\;x \in [0,\;1],\;\;} \hfill \\
 {u(0,t) =  u(1,t)=u_x (0,t) = u_x (1,t)=0, t\geq0,}\hfill \\
 {u(0,\;x)= u_0 (x), x \in [0,\;1],\;}\hfill \\
 \end{array} }} \right.
\]
Where $k $ is a real constant. We establish local well-posedness of
this closed-loop system by using Kato's theorem for abstract
quasilinear evolution equation of hyperbolic type. Then, by using
multiplier technique, we obtain a conservation law which enable us
to present a blow-up result.\\
\end{abstract}


\end{frontmatter}
\section{Introduction}
 \setcounter {equation}{0}Recently, Camassa and Holm \cite{1} derived a
nonlinear dispersive shallow water wave equation
\begin{equation}
\label{eq1.1} u_t - u_{xxt} + 3uu_x = 2u_x u_{xx} + uu_{xxx}
\end{equation}
\noindent which is called Camassa-Holm  equation. Here $u(x,t)$
denotes the fluid velocity at time $t$ in the $x$ direction or,
equivalently, the height of the water's free surface above a flat
bottom. Eq.(\ref{eq1.1}) has a bi-Hamiltonian structure  \cite{2,3}
and is completely integrable  \cite{1,4}. It admits, in addition to
smooth waves, a multitude of travelling wave solutions with
singularities: peakons, cuspons, stumpons and composite waves
\cite{1,5,6}. Its solitary waves are stable solitons \cite{7,8},
retaining their shape and form after interactions \cite{9}. It
models wave breaking \cite{10,11,12}.

The Cauchy problem for the Camassa-Holm equation has been studied
extensively. It has been proved to be locally well-posed
\cite{12,13} for initial data $u_0 \in H^S(R)$ with $S >
\textstyle{3 \over 2}$. Moreover, it has strong solutions that are
global in time \cite{14,15} as well as solutions that blow up in
finite time \cite{14,16,17,18}. On the other hand, it has global
weak solutions with initial data $u_0 \in H^1$ \cite{14,19,20}.

The initial boundary value problem for the Camassa-Holm equation was
also studied by several authors. For example, Kwek etc. \cite{21}
obtained the local existence and blow-up for an initial $^{ }$
boundary value problem for the Camassa-Holm equation with the
homogeneous boundary conditions: $u(0,t) = u_{xx}(0,t) = u(1,t) =
u_{xx}(1,t) = 0$ on interval $[0,1]$. Ma and Ding \cite{22} obtained
the existence and uniqueness of the local strong solutions to an
initial boundary problem for the Camassa-Holm equation on half axis
$R^ + $ with initial data $u_0 \in H^2(R^ + ) \cap H_0^1 (R^ + )$.
They also established the global result of the corresponding
solution, provided that the initial data $u_0 $ satisfies certain
positivity condition.

In this Letter we are interested in an initial boundary value
problem for the following equation
\begin{equation}
\label{eq1.2} u_t - u_{xxt} + 3uu_x - 2u_x u_{xx} - uu_{xxx} + k(u -
u_{xx} )_x = 0,
\end{equation}
where $k$ is a real constant, $ku_x$ denotes the dissipative term
and $ku_{xxx}$ denotes the dispersive effect. When $k = 0$,
Eq.(\ref{eq1.2}) is the well known Camassa-Holm equation. The
initial boundary value problem for Eq.(\ref{eq1.2}) we intend to
investigate is
\begin{equation}
\label{eq1.3} \left\{ {{\begin{array}{*{20}c}
 {u_t - u_{xxt} + 3uu_x - 2u_x u_{xx} - uu_{xxx} + k(u - u_{xx})_x=
0,  t\geq0,\;x \in [0,\;1],\;\;} \hfill \\
 {u(0,t) = u(1,t)=u_x (0,t) = u_x (1,t)=0, t\geq0,}\hfill \\
 {u(0,\;x) = u_0 (x), x \in [0,\;1],\;}\hfill \\
 \end{array} }} \right.
\end{equation}

The remainder of the paper is organized as follows. In Section 2, we
establish the local well-posedness for the closed-loop system
(\ref{eq1.3}) by Kato's theorem \cite{23}. In Section 3, by using
multiplier technique, we obtain a conservation law of the
closed-loop system (\ref{eq1.3}). Using this conservation law we
present a blow-up result.

We will use the following notation without further comment. $ * $
for convolution; $L(Y,X)$ for all bounded linear operator from
Banach space $Y$ to $X(L(X)$ if $X = Y)$; $\partial _x = \partial /
\partial x$; $\Lambda = (1 - \partial _x^2 )^{{1 \over
2}}$; $H^S $ is the usual Sobolev spaces, with the norm $\left\| {\;
\cdot \;} \right\|_{H^S} = \left\| {\; \cdot \;} \right\|_S $ and
the inner product $(\; \cdot \;)_S$; $L^2 = L^2(0,1)$ with the norm
$\left\| {\; \cdot \;} \right\|_0 $ and the inner product $(
\;\cdot\; )_0$; $H_{0,1}^S = \{u(x) \in H^S(0,1):\;u(0) =
u(1)=u_x(0) = u_x(1)=0\}$ ; $[A,\;B]=AB-BA$ denotes the commutator
of the linear operators $A$ and $B$; $C^k(I;\;X)$for the space of
all $k$ times continuously differentiable functions defined on an
interval $I$ with values in Banach space $X$.

\section{Local well-posedness}
 \setcounter {equation}{0}In this section , we
will apply Kato's theorem \cite{23} to esstablish the local
well-posedness for the closed-loop system (\ref{eq1.3}). For
convenience, we state Kato's theorem in the form suitable for our
purpose.

Consider the Cauchy problem associated to a quasilinear evolution
equation
\begin{equation}
\label{eq2.1} \quad \left\{ {\begin{array}{l}
 \frac{du}{dt} + A(u)u = f(u) \in X,\quad t \ge 0 \\
 u(0) = u_0 \in Y \\
 \end{array}} \right.
\end{equation}
\noindent where $A(u)$ is a linear operator depending on the unknown
$u$, and $u_0 $ the initial value. To study the Cauchy problem
(local in the time) associated to (\ref{eq2.1}), we will make the
following assumptions:

(X) $X$ and $Y$are reflexive Banach spaces where $Y \subset X$, with
the inclusion continuous and dense, and there is an isomorphism $Q$
from $Y$onto $X$.

($A_1$) Let $W$ be an open ball centered in 0 and contained in $Y$.
The linear operator $A(u)$ belongs to $G(X,1,\beta)$ where $\beta$
is a real number, i.e., $ - A(u)$ generates a $C_0$-semigroup such
that
\[
\left\| {e^{ - sA(u)}} \right\|_{L(X)} \le e^{\beta s}.
\]
Note that if $X$ is a Hilbert space, then $A \in G(X,1,\beta )$ if
and only if \cite{24}

(a) $(A\phi ,\phi )_X \ge - \beta \left\| \phi \right\|_X^2
,\;\forall \phi \in D(A)$,

(b) $(A + \lambda I )$is onto for some (all) $\lambda > \beta $.

Under these conditions $A(u)$ is said to be quasi-m-accretive.

($A_2$) The map $w \in W \to B(w) = [Q,\;A(w)]Q^{ - 1} \in L(X)$ is
uniformly bounded and Lipschitz continuous, that is, there exist
constants $\lambda _1 , \mu _1 > 0$, such that for all $w,\,y \in
W$,
\[
\left\| {B(w)} \right\|_{L(X)} \le \lambda _1 ,
\]
\[
\left\| {B(w)-B(y)} \right\|_{L(X)} \le \mu_1\|w-y\|_Y,
\]

($A_3$) $X \subseteq D(A(w))$ for each $w \in W$ (so that
$A(w)\left| {_Y \in L(Y,X)} \right.$by the Closed Graph theorem).
Moreover, the map $w \in W \to A(w) \in L(Y,X)$ satisfies the
following Lipschitz condition:
\[
\left\| {A(w) - A(y)} \right\|_{L(Y,X)} \le \mu _2 \left\| {w - y}
\right\|_X
\]
\noindent for all $w,\,y \in W$, where $\mu _2 $ is a non-negative
constant.

(f) The function $f:W \to Y$is bounded, i.e., there is a constant
$\lambda _2 > 0$ such that $\left\| {f(w)} \right\|_Y \le \lambda _2
$ for all $w \in W$,  and the function $w \in X \to f(w)$ is
Lipschitz in $X$(resp. in $Y)$, i.e.,
\[
\left\| {f(w) - f(y)} \right\|_X \le \mu _3 \left\| {w - y}
\right\|_X , \quad \forall w,\,y \in W,
\]
\[
\left\| {f(w) - f(y)} \right\|_Y \le \mu _4 \left\| {w - y}
\right\|_Y , \quad \forall w,\,y \in W,
\]
\noindent where $\mu _3 ,\;\mu _4 $ is non-negative constant.

We are now in position to state Kato's local well posedness result.
\begin{theorem}[Kato's theorem] Assume conditions (X); ($A_1$)-($A_3$) and (f) hold. Given $u_0 \in Y$, there
is $T > 0$ and unique solution $u \in C([0,T]; Y) \cap C^1([0,T];
X)$ to (\ref{eq2.1}) with $u(0) = u_0$. Moreover, the map $u_0 \in Y
\to u \in C([0,T]; Y)$ is continuous.
\end{theorem}

We now provide the framework in which we shall reformulate problem
(\ref{eq1.3}).

 Let $m = u - u_{xx}$ , then Eq.(\ref{eq1.2}) takes the form of a quasi-linear evolution
equation of hyperbolic type
\[
m_t + um_x + 2u_x m + km_x = 0.
\]
By using the operator $G(x) = \frac{\cosh (x - [x] - {1 \over
2})}{2\sinh ({1 \over 2})}$, where $[x]$ denotes the integer part of
$x \in [0,1]$, then$(1 - \partial _x^2 )^{ - 1}f = G
* f$, $\forall f \in L^2$ and $G * m = u$. Using this identity, we
can rewrite Eq.(\ref{eq1.2}) as
\[
u_t + uu_x + \partial _x (G * (u^2 +{1 \over 2}u_x^2 )) + ku_x = 0.
\]
Then the closed-loop system (\ref{eq1.3}) becomes
\begin{equation}
\label{eq2.2} \left\{ {{\begin{array}{*{20}c}
 {u_t + uu_x + \partial _x (G * (u^2 + {1 \over 2}u_x^2 )) + ku_x = 0, t \geq 0,\;x \in [0,\;1]\;{\kern 1pt}}
\hfill
\\
 {u(0,t) = u(1,t)=u_x (0,t) = u_x (1,t)=0, t \geq 0 \;\;\;\quad \;} \hfill \\
 {u(0,\;x) = u_0 (x), x \in [0,\;1]
\;}
\hfill \\
\end{array} }} \right.
\end{equation}
For the system (\ref{eq2.2}), we have the following result.
\begin{theorem}  Given $u_0 (x) \in H_{0,1}^2 $, there exists a maximal
value $T = T(u_0 (x)) > 0$ and a unique solution $u(x,t)$ to the
closed loop system (\ref{eq2.2}) such that $u = u( \cdot \;;\;u_0 )
\in C([0,T];H_{0,1}^2 ) \cap C^1([0,T];H^1)$. Moreover, the solution
depends continuously on the initial data, i.e., the mapping $u_0 \to
u( \cdot \;;\;u_0 ): H^2 \to C([0,T);H_{0,1}^2 ) \cap
C^1([0,T);H^1)$ is continuous.
\end{theorem}

Let $A(u)u = u\partial _x u + k\partial _x u $,  $f(u) = -
\partial _x (G * (u ^2 + \frac{1}{2}u_x^2 ))$,  $Q = \Lambda $, $X = H^1$, $Y = H_{0,1}^2 $.
\begin{remark}
The operator $Q$ maps $Y$ to $X$. In fact, $\forall u \in Y$,
\begin{equation}
\begin{split}
\left\| {Qu} \right\|_1^2 & = \left\| {Qu} \right\|_0^2 + \left\|
{(Qu)_x } \right\|_0^2\\
\nonumber &=(Qu,Qu)_0 + ((Qu)_x ,(Qu)_x )_0
\\ \nonumber
&=\int_0^1 {(u - u_{xx} } )udx + \int_0^1 {(u_x - u_{xxx} } )u_x
dx\\
\nonumber &=\int_0^1 {u^2} dx - uu_x \left| {_0^1 } \right. +
\int_0^1 {u_x^2 } dx + \int_0^1 {u_x^2 } dx - u_x u_{xx} \left|
{_0^1 } \right.
 + \int_0^1 {u_{xx}^2 } dx\\
\nonumber &=\left\| u \right\|_0^2 + 2\left\| {u_x } \right\|_0^2 +
\left\| {u_{xx} } \right\|_0^2,
\end{split}
\end{equation}
so $Qu \in X$. Moreover, the operator $Q$ is an isomorphism of $Y$
onto $X$. One may see the similar argument in \cite{28}.
\end{remark}

In order to prove Theorem 2.2, by applying Theorem 2.1, we only need
to verify $A(u)$ and $f(u)$ satisfy the conditions (X);
($A_1$)-($A_3$) and (f). The following lemmas are useful for our
arguments.
\begin{lemma}[\cite{25}] Let $f\in H^S$, $S> {3 \over 2}$, then
\[
\|\Lambda^{-r}[\Lambda^{r+t+1},M_f]\Lambda^{-t}\|_{L(L^2)}\leq
c\|f\|_S, |r|, |t|\leq S-1,
\]
where $M_f$ is the operator of multiplication by $f$ and $c$ is a
constant depending only on $r,t$.
\end{lemma}
\begin{lemma}[\cite{26}]
Let $X$ and $Y$ be two Banach spaces and $Y$ be continuously and
densely embedded in $X$. Let $-A$ be the infinitesimal generator of
the $C_0$-semigroup $T(t)$ on $X$ and $Q$ be an isomorphism from $Y$
onto $X$. $Y$ is $-A$-admissible (ie. $T(t)Y\subset Y, \forall t\geq
0$, and the restriction of  $T(t)$ to $Y$ is a $C_0$-semigroup on
$Y$.) if and only if $-A_1=-QAQ^{-1}$ is the infinitesimal generator
of the $C_0$-semigroup $T_1(t)=QT(t)Q^{-1}$ on $X$.
 Moreover, if $Y$ is $-A$-admissible, then the part of $-A$ in $Y$ is the infinitesimal generator of the restriction of  $T(t)$ to $Y$.
 \end{lemma}

Now we divide the proof of the Theorem 2.2 into the following
lemmas.
\begin{lemma}
The operator $A(u) = u\partial _x  + k\partial _x  $ with $u \in Y$
belongs to $G(L^2,1,\beta )$.
\end{lemma}
\begin{proof}
Due to $L^2$ being a Hilbert space, $A(u)\in G(L^2,1,\beta )$
\cite{24} if and only if there is a real number $\beta$ such that

(a)$(A(u)v,v)_0\geq -\beta \|v\|_0^2$,

(b)$-A(u)$ is the infinitesimal generator of a $C_0$-semigroup on
$L^2$ for some(or all) $\lambda \geq\beta$.

First, let us prove (a). Due to $u \in Y$, $u$ and $u_x$ belongs to
$L^\infty$. Noting that $\|u_x\|_{L^\infty}\leq \|u\|_2$, we have
\begin{equation}
\begin{split}
(A(u)v,v)_0 &=(uv_x+kv_x,v)_0\\
\nonumber &=(uv_x,v)_0+k(v_x,v)_0
\\ \nonumber
&={{1 \over 2}}(u_xv,v)_0+{k \over 2}v^2|_0^1\\
\nonumber &\leq {1 \over 2}\|u_x\|_{L^\infty}(v,v)_0\\
\nonumber &\leq{1 \over 2}\|u\|_2\|v\|_0^2
\end{split}
\end{equation}
Setting $\beta={1 \over 2}\|u\|_2$, we have $(A(u)v,v)_0\geq -\beta
\|v\|_0^2$.

Next we prove (b). Obviously $A(u) = u\partial _x+ k\partial _x$ is
a closed operator. For $w \in Y$, if there exists $z \in Y$, such
that $(\lambda I + A(u))^{ - 1}z = w$. Then $(\lambda I + A(u))w=z$,
i.e., $\lambda w + uw _x+kw_x = z$. Multiplying the both sides of
this equality by $w$, and integrating over $(0,1)$ by parts, we get
\begin{equation}
\label{eq2.4} \lambda \left\| w \right\|_X^2 -{1 \over 2}\int_0^1
{u_x } w^2dx = \int_0^1 z wdx.
\end{equation}
Since $u \in Y$, $u$ and $u_x $ belong to $L^\infty $. Note that
$\left\| {u_x } \right\|_{L^\infty } \le \left\| u \right\|_Y $,
then we have
\begin{equation}
\label{eq2.5} {1 \over 2}\int_0^1 {u_x } w^2dx \le {1 \over
2}\left\| {u_x } \right\|_{L^\infty } \left\| w \right\|_X^2\le{1
\over 2}\left\| u \right\|_Y\left\| w \right\|_X^2.
\end{equation}
On the other hand, by H\"{o}lder inequality, we have
\begin{equation}
\label{eq2.6} \int_0^1 z wdx \le \left\| z \right\|_X \left\| w
\right\|_X.
\end{equation}
Then combining (\ref{eq2.4})-(\ref{eq2.6}), we obtain
\[
\lambda \left\| w \right\|_X^2 -{1 \over 2}\left\|u \right\|_Y
\left\| w \right\|_X^2 \le \left\| z \right\|_X \left\| w \right\|_X
.
\]
That is
\begin{equation}
\label{eq2.7} (\lambda - {1 \over 2}\left\| u \right\|_Y)\left\| w
\right\|_X^2 \le \left\| z \right\|_X \left\| w \right\|_X.
\end{equation}
Setting $\beta = {1 \over 2}\left\| v \right\|_Y $, then by
(\ref{eq2.7}), we have
\[
\left\| w \right\|_X \le \frac{1}{\lambda - \beta }\left\| z
\right\|_X ,\quad \forall \lambda > \beta.
\]
i.e.,
\[
\left\|(\lambda I + A(u))^{ - 1}\right\|_{L(X)}\le \frac{1}{\lambda
- \beta },\quad \forall \lambda > \beta.
\]
By Hille-Yosida theorem \cite{27}, we conclude that the operator
$-A(u)$ is the infinitesimal generator of a $C_0$-semigroup on $X$.
This completes the proof of Lemma 2.3.
\end{proof}
\noindent
\begin{lemma}
 The operator $A(u) = u\partial _x  + k\partial _x  $ with
$u \in Y$ belongs to $G(H^1,1,\beta )$.
\end{lemma}
\begin{proof}
Due to $H^1$ being a Hilbert space, $A(u)\in G(H^1,1,\beta )$
\cite{24} if and only if there is a real number $\beta$ such that

(a)$(A(u)v,v)_1\geq -\beta \|v\|_1^2$,

(b)$-A(u)$ is the infinitesimal generator of a $C_0$-semigroup on
$H^1$ for some(or all) $\lambda \geq\beta$.

 First, let us prove (a). Let $u \in Y$, then $u$ and $u_x$
belongs to $L^\infty$ and $\|u_x\|_{L^\infty}\leq \|u\|_2$. Note
that
\[
\Lambda(uv_x)=[\Lambda, u]v_x+u\Lambda v_x=[\Lambda,
u]v_x+u\partial_x\Lambda v.
\]
Then we have
\begin{equation}
\begin{split}
(A(u)v,v)_1&=(\Lambda (uv_x+kv_x),\Lambda v)_0\\ \nonumber
 &=([\Lambda,u]v_x,\Lambda v)_0-{{{1} \over {2}}}(u_x \Lambda v, \Lambda v)_0+k(\Lambda v_x,\Lambda v)_0
\\ \nonumber
&\leq \| [\Lambda , u]\|_{L(L^2)}\|\Lambda v \|_0^2+
{{1 \over 2}}\|u_x\|_{L^\infty}\|\Lambda v\|_0^2\\
\nonumber &\leq c\|u\|_2\| v\|_1^2,
\end{split}
\end{equation}
where we apply Lemma 2.1 with $r=0$ , $t=0$. Setting
$\beta=c\|u\|_2$, we have $(A(u)v,v)_1\geq -\beta \|v\|_1^2$.

Next we prove (b). Note that $Q=\Lambda$ is an isomorphism of $Y$
onto $X$ and $Y$ is continuously and densely embeded in $X$. Define
\[
A_1(u):=[Q,A(u)]Q^{-1}=[\Lambda, A(u)]\Lambda^{-1}, \quad
B_1(u):=A_1(u)-A(u),
\]
then
\begin{equation}
\begin{split}
 \label{eq2.26}
\nonumber B_1 (u) &= [\Lambda,\;(u + \gamma )\partial _x ]\Lambda
^{-1} - (u + \gamma )\partial _x \\ & \nonumber =
[\Lambda,\;u\partial _x ]\Lambda ^{-1} + \gamma \Lambda \partial _x
\Lambda ^{-1} - (u + \gamma )\partial _x
\\ &
\nonumber = [\Lambda,\;u]\partial _x \Lambda ^{-1} + u\Lambda
\partial _x \Lambda ^{-1} - u\partial _x \\ & \nonumber =
[\Lambda,\;u]\partial _x \Lambda ^{-1}.
\end{split}
\end{equation}
 Let $v\in L^2$ and $u\in Y$. Then we have

\begin{equation}
\begin{split}
 \label{eq2.27}
\nonumber \left\| {B_1 (u)v} \right\|_0 &= \left\| {[\Lambda
,\;u]\partial _x \Lambda ^{-1}v} \right\|_0 \\ &
 \nonumber = \left\| {[\Lambda,\;u]\Lambda ^{ - 1}\partial _x
v} \right\|_0 \\ & \nonumber
 \le \left\| {[\Lambda,\;u]}\right\|_{L(L^2)}
\left\| {\Lambda ^{ - 1}\partial _x v}\right\|_0 \\ &
\nonumber \le c\left\| u \right\|_2 \left\| v \right\|_0 \\
\end{split}
\end{equation}

where we apply Lemma 2.1 with $r=0$, $t=0$ . Therefore, we obtain
that $B_1(u)\in L(L^2)$.

Note that $A_1(u)=A(u)+B_1(u)$ and $A(u)\in G(L^2,1,\beta)$ in Lemma
2.3. By a perturbation theorem for semigroups (cf. \S  5.2 Theorem
2.3 in \cite{26}), we obtain that $A_1(u)\in G(L^2,1,\beta')$. By
applying Lemma 2.2 with $Y=H_{0,1}^1$, $X=L^2$ and $Q=\Lambda$, we
conclude that $Y$ is $-A$-admissible. So, $-A(u)$ is the
infinitesimal generator of a $C_0$-semigroup on $Y$. This completes
the proof of Lemma 2.4.
\end{proof}
\begin{lemma}
For all $u\in Y$, $A(u)\in L(Y,X)$. Moreover,
\[
\left\| (A(u)-A(z))w \right\|_X\leq \left\| u-z\right\|_X\left\| w
\right\|_Y, u,z,w\in Y.
\]
\end{lemma}
\begin{proof}
For  $u,z,w\in Y$, we have
\begin{equation}
\begin{split}
\left\| (A(u)-A(z))w \right\|_X &=\left\| (uw_x+kw_x)-(zw_x+kw_x)
\right\|_X \\ \nonumber & \leq \left\| u-z\right\|_X\left\|w_x
\right\|_{L^\infty} \\ \nonumber &\leq\left\| u-z \right\|_X\left\|w
\right\|_Y.
\end{split}
\end{equation}
 Taking $z=0$ in the above inequality, we obtain $A(u)\in
L(Y,X)$. This completes the proof of Lemma 2.5.
\end{proof}
\begin{lemma}
For $u\in Y$, $B(u)=[\Lambda, A(u)]\Lambda^{ - 1}\in L(X)$, and
\[
\left\| {(B(u) - B(z))w} \right\|_X \le c\left\| {u - z} \right\|_Y
\left\| w \right\|_X , \quad \forall  u,z \in Y, \quad w \in X.
\]
\end{lemma}
\begin{proof}
Let $u,z \in Y$, $w \in X$, then
\begin{equation}
\begin{split}
 \left\| {(B(u) - B(z))w} \right\|_X &= \left\| \Lambda{[\Lambda
,(u - z)\partial _x ]\Lambda ^{ - 1}w} \right\|_0\\ \nonumber &
 \le \left\| \Lambda{[\Lambda ,(u - z)]} \Lambda ^{ - 1}\right\|_{L(L^2)} \left\|
{w_x } \right\|_0 \\ \nonumber & \le c \left\| {u - z} \right\|_Y
\left\| w \right\|_X ,
\end{split}
\end{equation}
\noindent where we apply the Lemma 2.1 with $r=-1$ and $t=1$ .
Taking $z= 0$ in the above inequality, we obtain $B(u) \in L(X)$ .
This completes the proof of Lemma 2.6.
\end{proof}
\begin{lemma}
$f(u) = -
\partial _x (G * (u^2 + {1 \over 2}u_x^2 ))$ satisfies

(a) $\left\| {f(u)} \right\|_Y \le c\left\| u \right\|_Y^2 $, $u \in
Y$;

(b) $\left\| {f(u) - f(z)} \right\|_X \le c\left\| {u - z}
\right\|_X $, $u,z \in Y$;

(c) $\left\| {f(u) - f(z)} \right\|_Y \le c\left\| {u - z}
\right\|_Y $, $u,z \in Y$.
\end{lemma}
\begin{proof}
Let $u,z \in Y$. Note that $H_{0,1}^1$ is a Banach algebra. Then we
have
\begin{equation}
\begin{split}
\left\| {f(u) - f(z)} \right\|_Y &= \| -\partial _x G*( u^2-z^2 + {1 \over 2}u_x^2 -{1 \over 2}z_x^2 )\|_2\\
\nonumber & \le \|(u - z)(u + z)\|_1+{1\over2}\|(u-z)_x(u+z)_x\|_1
\\ \nonumber
& \leq \| u - z\|_1 \|u + z \|_1+{1\over
2}\|(u-z)_x\|_1\|(u+z)_x\|_1\\ \nonumber & \leq
\|u-z\|_2\|u+z\|_2+{1\over 2}\|u-z\|_2\|u+z\|_2\\
\nonumber & \leq ({3\over 2}\|u\|_2+{3\over 2}\|z\|_2)\|u-z\|_2\\
\nonumber & =({3\over 2}\|u\|_Y+{3\over 2}\|z\|_Y)\|u-z\|_Y.
\end{split}
\end{equation}
This proves (c). Taking $z = 0$ in the above inequality, we obtain
(a). Next, we prove (b).

Let $v,z\in Y$. Note that $H_{0,1}^1 $ is a Banach algebra. Then we
get
\begin{equation}
\begin{split}
\left\| {f(u) - f(z)} \right\|_X &= \| -\partial _x G*( u^2-z^2 + {1 \over 2}u_x^2 -{1 \over 2}z_x^2 )\|_1\\
\nonumber & \le \|(u - z)(u + z)\|_0+{1\over2}\|(u-z)_x(u+z)_x\|_0
\\ \nonumber &
 \leq \| u - z\|_0 \|u + z \|_0+{1\over
2}\|(u-z)_x\|_0\|(u+z)_x\|_0\\ \nonumber & \leq
\|u-z\|_1\|u+z\|_2+{1\over 2}\|u-z\|_1\|u+z\|_2\\
\nonumber & \leq ({3\over 2}\|u\|_2+{3\over 2}\|z\|_2)\|u-z\|_1\\
\nonumber & =({3\over 2}\|u\|_Y+{3\over 2}\|z\|_Y)\|u-z\|_X.
\end{split}
\end{equation}
 This completes the proof of Lemma 2.7.
\end{proof}\\
\begin{proof}[Proof of Theorem 2.2] Combining Theorem 2.1 and Lemma 2.3-2.7, we get the statement of Theorem 2.2.
\end{proof}

\section{Blow up}
 \setcounter {equation}{0}
    Firstly, by using multiplier technique, we obtain the following
conservation law of the closed-loop system (\ref{eq2.2}).
\begin{theorem}
Let $T >0$ be the maximal time of existence of the solution $u(x,t)$
to the closed-loop system(\ref{eq2.2}) (or (\ref{eq1.3})) with the
initial data $u_0 (x) \in H_{0,1}^2 $, then
\begin{equation}
\label{eq3.1}\left\| {u( \cdot ,t)} \right\|_1^2 = \left\| {u_0}
\right\|_1^2 .
\end{equation}
\end{theorem}
\begin{proof}
 Multiplying the first equation of the closed-loop system (\ref{eq2.2}) by $u -
u_{xx} $, and integrating over $(0,1)$ by parts, we get

\begin{equation}
\label{eq3.2}
\begin{split}
 \int_0^1 {u_t (u - u_{xx} )dx}   &= - \int_0^1 {uu_x (u - u_{xx} )dx} -k\int_0^1 {u_x (u
-u_{xx} )dx}\\
&- \int_0^1 {\partial _x (G * (u^2 + {1 \over 2}u_x^2 ))(u - u_{xx}
)dx}
\end{split}
\end{equation}
For the LHS of (\ref{eq3.2}), using the  boundary conditions of the
closed-loop system (\ref{eq2.2}), we have
\begin{equation}
\label{eq3.3}
\begin{split}
\int_0^1 {u_t (u - u_{xx} )dx}  &
 = {1 \over 2}{d \over {dt}}\int_0^1 {u^2dx - u_t u_x
\left| {_0^1 } \right. + {1 \over 2}} {d \over {dt}}\int_0^1 {u_x^2
dx}\\ &
 ={1 \over 2}{d \over {dt}}(\int_0^1 {u^2dx
+ } \int_0^1 {u_x^2 dx} )
\end{split}
\end{equation}
For the RHS of (\ref{eq3.2}), using the  boundary conditions of the
closed-loop system (\ref{eq2.2}), we obtain
\begin{equation}
\label{eq3.4}
\begin{split}
 &- \int_0^1 {uu_x (u - u_{xx} )dx} - k\int_0^1 {u_x (u - u_{xx} )dx} -
\int_0^1 {\partial _x (G * (u^2 + {1 \over 2}u_x^2 ))(u - u_{xx}
)dx}
\\&
 = - \int_0^1 {u^2u_x dx} + \int_0^1 {uu_x u_{xx} dx} - k\int_0^1 {uu_x dx +
k} \int_0^1 {u_x u_{xx} dx}
\\&
 - G * (u^2 + {1 \over 2}u_x^2 )u\left| {_0^1 } \right.  +
\int_0^1 {G * (u^2 + {1 \over 2}u_x^2 )u_x dx}  \\& +\partial _x (G
* (u^2 + {1 \over 2}u_x^2 ))u_x \left| {_0^1 } \right. -
\int_0^1 {\partial _x^2 (G * (u^2 + {1 \over 2}u_x^2 ))} u_x dx
\\&
 = - {1 \over 3}u^3\left| {_0^1 } \right. + \int_0^1 {uu_x u_{xx}
dx} - {k \over 2}u^2\left| {_0^1 } \right. + {k \over 2}u_x^2 \left|
{_0^1 } \right. + \int_0^1 {(u^2 + {1 \over 2}u_x^2 )u_x dx}
\\&
 = \int_0^1 {uu_x u_{xx} dx} + (u^2 + {1 \over 2}u_x^2
)u\left| {_0^1 } \right. - \int_0^1 {(2uu_x + u_x u_{xx} )udx}
\\&
 = - \int_0^1 {2u^2u_x dx}  = -{2 \over
3}u^3\left| {_0^1 } \right. = 0
\end{split}
\end{equation}
It follows from (\ref{eq3.2})-(\ref{eq3.4}) that
\begin{equation}
\label{eq3.5}{1 \over 2}{d \over {dt}}\int_0^1 {(u^2 + u_x^2 )dx} =
0
\end{equation}
Integrating (\ref{eq3.5}) over $(0,t)$, we obtain (\ref{eq3.1}).
\end{proof}
\begin{remark}  Employing the Agmon's inequality
\[
\int_0^1{u^2(x)}dx\leq 2u^2(0)+4\int_0^1{u_x^2(x)} dx \] and the
Poincar\'{e} inequality \[ \max_{x\in[0,1]}u^2(x)\leq
u^2(0)+2\sqrt{\int_0^1{u^2(x)}dx}\sqrt{\int_0^1{u_x^2(x)}dx}\] with
$u(0)=0$, we have
\[ \max_{x\in[0,1]}{|u(x)|}\leq 2\sqrt{\int_0^1{u_x^2(x)}dx}\]
It then  follows from Theorem 3.1  that
\[ \sup_{(x,t)\in[0,1]\times[0,+\infty)}|u(x,t)| \leq 2\|u_0\|_1
\]
This shows that the solution $u(x,t)$ to the closed-loop system
(\ref{eq2.2}) or system (\ref{eq1.3}) is bounded if the initial data
$u_0 (x) \in H_{0,1}^2 $.
\end{remark}

 Now we present a blow-up result of solution to the closed-loop system (\ref{eq2.2}) (or
system (\ref{eq1.3})).
\begin{theorem}
Assume $u_0 (x) \in H_{0,1}^2 $ and $T$ is the maximal existence
time of the solution $u(x,t)$ to the closed-loop system
(\ref{eq2.2})(or (\ref{eq1.3})) guaranteed by Theorem 2.2. If there
exists one point $x_0\in (0,1)$ such that $u_{xx}(x_0,t)=0$ and
$u'_0(x_0)<-\sqrt{2}\|u_0\|_1$, then the corresponding solution
blows up in finite time. Moreover, the maximal time of existence is
estimated above by
\[
{1 \over \sqrt{2}\|u_0\|_1} \ln ({h(0)-\sqrt{2}\|u_0\|_1 \over h(0)
+\sqrt{2}\|u_0\|_1})
\]
where $h(0)=u'_0(x_0)$.
\end{theorem}

\begin{proof}  Differentiating the first
equation of the closed-loop system (\ref{eq2.2})  with respect to
$x$, in view of $\partial_x^2G*f=G*f-f$, we have
\begin{equation}
\label{eq3.6} u_{tx}=-{1 \over 2}u_x^2-uu_{xx}+\gamma
u_{xx}+u^2-G*(u^2+{1 \over 2}u_x^2)
\end{equation}
Let $x=x_0$ in (\ref{eq3.6}) and set $h(t)=u_x(x_0,t)$. Noting that
$u_{xx}(x_0,t)=0$ and $G*(u^2+{1 \over 2}u_x^2)\geq 0$, we obtain

\begin{equation}
\label{eq3.7} h'(t)\leq -{1 \over 2}h^2(t)+u^2(x_0,t).
\end{equation}
In view of (\ref{eq3.1}) and Sobolev embedding theorem, we have
\begin{equation}
\label{eq3.8} u^2(x_0,t)\leq \|u\|_{L^\infty}^2\leq \|u\|_1^2\ =
\|u_0\|_1^2.
\end{equation}
It follows from (\ref{eq3.7}) and (\ref{eq3.8}) that
\begin{equation}
 h'(t)\leq -{1 \over 2}h^2(t)+\|u_0\|_1^2.\nonumber
\end{equation}
Note that if $h(0)\leq -\sqrt{2}\|u_0\|_1$, then $h(t)\leq
\sqrt{2}\|u_0\|_1$, for all $t \in [0,T )$. Therefore, from the
above inequality we obtain
\[
{h(0) +\sqrt{2}\|u_0\|_1 \over h(0)
-\sqrt{2}\|u_0\|_1}e^{\sqrt{2}\|u_0\|_1t}-1\leq {2\sqrt{2}\|u_0\|_1
\over h(t)-\sqrt{2}\|u_0\|_1}\leq 0.
\]
Due to $0<{h(0) +\sqrt{2}\|u_0\|_1 \over h(0)
-\sqrt{2}\|u_0\|_1}<1$, then exists
\[
T_0 \leq {1 \over \sqrt{2}\|u_0\|_1} \ln ({h(0)-\sqrt{2}\|u_0\|_1
\over h(0) +\sqrt{2}\|u_0\|_1})
\]
such that
\[
\lim_{t\rightarrow T_0} h(t)=-\infty.
\]
Thus $\lim_{t\rightarrow T_0} \|u\|_2=\infty$ because
$|u_x(x,t)|\leq\|u_x(x,t)\|_{L^\infty}\leq\|u\|_2$. That is, the
solution $u(x,t)$ to the closed-loop system (\ref{eq2.2}) does not
exist globally in time in function space $H_{0,1}^2$.
\end{proof}
\begin{remark}  Theorem 3.2 shows that although  $\int_0^1{u_x^2}dx$
is bounded (see Theorem 3.1) , it does not guarantee that $u_x(x,t)$
is bounded for all $x\in[0,1]$.
\end{remark}


\begin{thebibliography}{00}
\bibitem{1}R. Camassa, D. D. Holm,  Phys. Rev. Lett. 71 (1993) 1661.
\bibitem{2}A. Fokas, B. Fuchssteiner, Physica D  4
(1981)  47.
\bibitem{3}J. Lenells,  J. Phys. A  38
(2005)  869.
\bibitem{4}A. Constantin,  Proc. R. Soc. London A  457 (2001) 953.
\bibitem{5}J. Lenells, J. Differential
Equations  217 (2005) 393.
\bibitem{6}Y. Li, P. Olver,  J. Differential
Equations  162 (2000) 27.
\bibitem{7}A. Constantin, W. A. Strauss,  J. Nonlinear
Sci.  12 (2002)  415.
\bibitem{8}A. Constantin, W. A. Strauss,  Commun. Pure Appl.
Math.  53 (2000) 603.
\bibitem{9} R. S. Johnson,  Proc. R. Soc. London A  459
(2003) 1687.
\bibitem{10}A . Constantin,  Ann. Inst. Fourier
(Grenoble) 50 (2000) 321.
\bibitem{11}A. Constantin, J. Escher,  Acta Math.  181 (1998)  229.
\bibitem{12}A. Constantin, J. Escher,  Math. Z.  233
(2000) 75.
\bibitem{13}G. Rodriguez-Blanco,  Nonlinear Anal.  46 (2001)  309.
\bibitem{14}  A. Constantin, J. Escher,  Comm. Pure Appl. Math.   51 (1998)  475.
\bibitem{15}  H. Dai, K. Kwek, H. Gao, C. Qu,  Front. Math. China   1 (2006),
144.
\bibitem{16}  A. Constantin,  J. Differential Equations  141 (1997) 218.
\bibitem{17}A. Constantin,  J. Nonlinear Sci. 10 (2000)  391.
\bibitem{18}A. Constantin, J. Escher,  Ann. Sci. Norm. Sup. Pisa   26
(1998) 303.
\bibitem{19} A. Constantin, L. Molinet,  Comm. Math. Phys.  211 (2000)  45 .
\bibitem{20}Z. Xin, P. Zhang,  Comm. Pure Appl. Math.   53
(2000)  1411.
\bibitem{21}K. Kwek, H. Gao, W. Zhang, C. Qu,  J. Math. Phys.   41 (2000) 8279 .
\bibitem{22}S. Ma, S. Ding,  J. Math. Phys.   45 (2004)  3479.
\bibitem{23}T. Kato: Quasi-linear equations of evolution, with applications to
partial differential quations. In:  Spectral Theory and Differential
Equations,  Lecture Notes in Math. 1975, pp. 25-70.
\bibitem{24}T. Kato, Adv. Math. Suppl. Stud.  8 (1983)  93.
\bibitem{25}T. Kato, Manuscripta
Math.  28 (1979) 89.
\bibitem{26}A. Pazy, Semigroup of Linear Operators and
Applications to Partial Differential Equations, Springer Verlag, New
York, 1983.
\bibitem{27}K. Yosida: Functional Analysis, Berlin/New York,
Springer-Verlag, 1966.
\bibitem{28}X. Zong, Y. Zhao, Nonlinear Analysis 67 (2007) 3167.

\end{thebibliography}
\end{document}